\newcommand*{\preliminary}{}

\documentclass{article}

\ifdefined\preliminary
    \ifdefined\quomodocumque
        \newfont{\domofont}{/domo}
        
    \fi
    \newcommand{\Scomment}[1]{{\color{Violet} \fbox{\parbox{\dimexpr\linewidth-2\fboxsep-2\fboxrule}{ Sarah:  #1}}}}
    \newcommand{\Rcomment}[1]{{\color{Green}  \fbox{\parbox{\dimexpr\linewidth-2\fboxsep-2\fboxrule}{  Ryan:  #1}}}}
\else

    \newcommand{\Scomment}[1]{}
    \newcommand{\Rcomment}[1]{}

\fi

\usepackage{sarah}
\usepackage{odonnell}
\usepackage{mathrsfs}
\usetikzlibrary{arrows,automata, positioning, decorations.pathreplacing, calc}
\usepackage{cite}
\usepackage{hyperref}

\newcommand{\avgCov}{\mathrm{avgCov}}
\newcommand{\avgInfo}{\mathrm{avgInfo}}
\newcommand{\avgCovCond}[1]{\avgCov_{\mid #1}}
\newcommand{\avgInfoCond}[1]{\avgInfo_{\mid #1}}

\begin{document}

\title{Conditioning and covariance on caterpillars}

\author{Sarah R. Allen\thanks{
                     Department of Computer Science, Carnegie Mellon University.
                     This material is based upon work supported by the National Science Foundation Graduate Research Fellowship Program under Grant No. 0946825.
                    \texttt{srallen@cs.cmu.edu}}
   \and Ryan O'Donnell\thanks{
                    Department of Computer Science, Carnegie Mellon University.
                    Supported by NSF grants CCF-0747250 and CCF-1116594.  Part of this work performed at the Bo\u{g}azi\c{c}i University Computer Engineering Department, supported by Marie Curie International Incoming Fellowship project number 626373.
                    \texttt{odonnell@cs.cmu.edu}}}

\maketitle

\begin{abstract}
    Let $\bX_1, \dots, \bX_n$ be joint $\{ \pm 1\}$-valued random variables.  It is known that conditioning on a random subset of~$O(1/\eps^2)$ of them reduces their average pairwise covariance to below~$\eps$ (in expectation).  We conjecture that $O(1/\eps^2)$ can be improved to~$O(1/\eps)$.  The motivation for the problem and our conjectured improvement comes from the theory of global correlation rounding for convex relaxation hierarchies.  We suggest attempting the conjecture in the case that $\bX_1, \dots, \bX_n$ are the leaves of an information flow tree. We prove the conjecture in the case that the information flow tree is a caterpillar graph (similar to a two-state hidden Markov model).
\end{abstract}

\section{Introduction}                                  \label{sec:intro}
Let $\bX = (\bX_1, \dots, \bX_n)$ be a list of jointly distributed Boolean random variables taking values in~$\{\pm 1\}$.  We are interested in the quantity
\[
    \avg_{\substack{\text{distinct pairs} \\ u, v \in [n]}} \SET{\bigl \lvert \Cov[\bX_u,\bX_v] \bigr \rvert} \in [0,1].
\]
For brevity we call this the \emph{average covariance} of the random variables (absolute-value sign notwithstanding).  It is a quantification of the extent to which the random variables are (pairwise) independent.

If the average covariance of $\bX_1, \dots, \bX_n$ is not small, then in some sense a ``typical''~$\bX_j$ contains a considerable amount of information about a sizeable fraction of the other~$\bX_k$'s.  Then if we condition on~$\bX_j$, we might expect the variance of these other $\bX_k$'s to decrease, thereby decreasing the overall average covariance.   For $t \in \Z^+$, we introduce the following notation:
\[
    \avgCovCond{t}(\bX) \coloneqq \avg_{\substack{J \subseteq [n] \\ \abs{J} = t}} \ \ \avg_{\substack{\text{distinct pairs} \\ u, v \in [n] \setminus J}} \SET{ \E\Brak{\Abs{\Cov[\bX_u,\bX_v]} \mid (\bX_j)_{j \in J}}}.
\]
The intuitions just described lead to the idea that choosing large~$t$ should cause $\avgCovCond{t}(\bX)$ to become small. Indeed, the following has recently been proven~\cite{Guruswami-Sinop,Barak-Raghavendra-Steurer,Raghavendra-Tan}:
\begin{theorem}                                     \label{thm:r-t}
    Let $\bX = (\bX_1, \dots, \bX_n)$ be $\{\pm 1\}$-valued random variables and let $0 < \eps \leq 1$.  Then for some integer $0 \leq t \leq O(1/\eps^2)$ it holds that $\avgCovCond{t}(\bX) \leq \eps$.
\end{theorem}

We present the following conjecture, made jointly with Yuan Zhou.
\begin{named}{Conjecture A}
    Theorem~\ref{thm:r-t} holds with $O(1/\eps)$ in place of $O(1/\eps^2)$.
\end{named}
\begin{remark}                                      \label{rem:order-of-quantifiers}
    In Theorem~\ref{thm:r-t}, by $t \leq O(1/\eps^2)$ we mean $t \leq C/\eps^2$ where~$C$ is a universal constant independent of $\bX_1, \dots, \bX_n$.  (We also assume $n \geq C/\eps^2 + 2$.) However one \emph{cannot} simply fix $t = \lceil C/\eps^2 \rceil$ independently of $\bX_1, \dots, \bX_n$; this would make Theorem~\ref{thm:r-t} false (see Proposition~\ref{prop:annoying-counterexample}).  These comments apply equally to Conjecture~A with $O(1/\eps)$ in place of $O(1/\eps^2)$.
\end{remark}

Motivation for Theorem~\ref{thm:r-t} and Conjecture~A comes from the theory of rounding algorithms for convex relaxations of optimization problems; specifically, the ``correlation rounding'' technique for the Sherali--Adams and SOS hierarchies.  In Section~\ref{sec:motivation} we further discuss this motivation, as well as the importance of improving the bound $t \leq O(1/\eps^2)$ to $t \leq O(1/\eps)$.

We were led to make Conjecture~A based on algorithmic optimism as well as being unable to find any counterexample refuting it.  The following example (which we call the ``homogeneous star'') is particularly instructive.  Let $\bX_0 \sim \{\pm 1\}$ be uniformly random and suppose $\bX = (\bX_1, \dots, \bX_n)$ is a list of independent ``$\rho$-correlated'' copies of~$\bX_0$ (where $\rho \in [0,1]$). I.e., for each $j \in [n]$ we have $\bX_j = \bX_0 \bR_j$, where $\bR_1, \dots, \bR_n$ are independent $\{\pm 1\}$-valued random variables satisfying $\E[\bR_j] = \rho$. By symmetry, all sets~$J$ in the definition of $\avgCovCond{t}(\bX)$ contribute equally to the average, so suppose we condition on~$\bX_1, \dots, \bX_t$.  It is not hard to check that the conditional average covariance of $\bX_{t+1}, \dots, \bX_n$ is then
\[
    \rho^2 \Var[\bX_0 \mid \bX_1, \dots, \bX_t].
\]
If $\rho \leq \sqrt{\eps}$ then this quantity is automatically at most~$\eps$, even without conditioning.  On the other hand, if $\rho \gg \sqrt{\eps}$ then we need to rely on the conditional variance above being small.  It's not difficult to show via a Hoeffding bound that this conditional variance is very small if (and only if)~$t \rho^2 \gg 1$; i.e., $\rho \gg 1/\sqrt{t}$.  Thus by taking~$t$ a little bigger than~$1/\eps$, the case of $\rho \gg \sqrt{\eps}$ is handled as well.  In other words, these rough calculations confirm (perhaps up to a log factor) that Conjecture~A holds for the homogeneous star for every value of~$\rho$.  On the other hand, this example also implies that one cannot hope for an improved bound of $t < o(1/\eps)$ in Conjecture~A.
\subsection{Information flow trees}
Being unable to prove Conjecture~A, we turn to trying to prove it in a wide family of special cases.  Specifically, we study the conjecture in the special case of \emph{information flow trees} (which includes the homogeneous star example discussed above).  Information flow trees have been studied in an extremely wide variety of contexts, under various names: in the theory of noisy communication and computation; in statistical physics (as the \emph{Ising model} on trees); in biology (as \emph{phylogenetic trees}); and in learning theory (as Markov networks/graphical models).  See Evans et al.~\cite{evans} for a number of results, and Mossel~\cite{mossel} for a survey.
\begin{definition}
    An \emph{information flow tree} $\calT = (V, E, \rho)$ is an undirected tree graph~$(V,E)$ (with $|V| > 1$) together with a function $\rho \co E \to [-1,1]$ giving a \emph{correlation} parameter for each edge.  We think of $\calT$ as generating a collection of $\{\pm 1\}$-valued random variables $(\bX_v)_{v \in V}$, $(\bR_e)_{e \in E}$ as follows:  First, the random variables~$\bR_e \in \{\pm 1\}$ are chosen such that $\E[\bR_e] = \rho(e)$, independently for all $e \in E$.  Next, the random variables $(\bX_v)_{v \in V}$ are collectively chosen so that $\bX_u \bX_v = \bR_{(u,v)}$ holds for all $(u,v) \in E$, uniformly at random from the two possibilities.
\end{definition}
\begin{remark}                                                  \label{rem:tree-def}
    An equivalent way to think of the $(\bX_v)$ random variables being generated is as follows: First, a vertex $r \in V$ is chosen to be the ``root''.  (This choice can be arbitrary, as it does not affect the final distribution.)  Next, $\bX_r$ is chosen uniformly at random from~$\{\pm 1\}$.  Finally, the remaining random variables $(\bX_v)_{v \neq r}$ are determined by ``noisily propagating'' $\bX_r$'s value along edges of the tree:  if~$\bX_u$ has been chosen, and $(u,v) \in E$, then $\bX_v$ is set to $\bX_u$ with probability $\half + \half \rho(u,v)$ and is set to $-\bX_u$ otherwise.  We add the remark that in the end, each $\bX_v$ is individually uniformly distributed on~$\{\pm 1\}$.
\end{remark}
\begin{remark}
    When discussing information flow trees, we think of the vertex random variables~$\bX_v$ as the main objects of interest, and the edge random variables~$\bR_e$ merely as ancillary information used to construct the~$\bX_v$'s.  Furthermore, if $V = L \sqcup M$ is the partition of~$V$ into \emph{leaf} vertices~$L$ and \emph{internal} vertices~$M$, we usually think of the \emph{leaf random variables} $(\bX_v)_{v \in L}$ as being ``observable'' and the internal random variables $(\bX_v)_{v \in M}$ as being ``hidden''.
\end{remark}

In this paper we study the special case of Conjecture~A in which $\bX_1, \dots, \bX_n$ are the leaf random variables of an information flow tree. Referring to Remark~\ref{rem:order-of-quantifiers}, in this case we conjecture it \emph{is} possible to fix $t = \text{const}/\eps$ independently of~$\bX_1, \dots, \bX_n$.  Assuming we can fix~$t$ allows us to make a few more simplifications.  Since
\[
    \avgCovCond{t}(\bX) = \avg_{\substack{U \subseteq [n] \\ |U| = t+2}} \SET{\avgCovCond{t}\Paren{(\bX_k)_{k \in U}}},
\]
it follows that proving the conjecture in the $n = t+2$ case suffices to prove it for general $n \geq t+2$.  And when $n = t+2$, the experiment reduces to the following:  we choose a random pair of leaves $u$ and $v$, condition on \emph{all} other leaf random variables~$\bX_w$, and then measure the (conditional) covariance of $\bX_u, \bX_v$.  Thus we are led to the following conjecture (in which we write~$t$ instead of~$t+2$ for notational simplicity):
\begin{named}{Conjecture B}
    Let $\calT$ be an information flow tree with leaf random variables $\bX_1, \dots, \bX_t$ (where $t \geq 2$).  Then
    \begin{equation}                            \label{eqn:conj-b}
        \avg_{\substack{\text{distinct pairs} \\ u, v \in [t]}} \E\Brak{\Abs{\Cov[\bX_u,\bX_v]} \mid (\bX_j)_{j \in [t] \setminus \{u,v\}}} \leq O(1/t).
    \end{equation}
\end{named}
\noindent We emphasize that Conjecture~B implies Conjecture~A in the case that $\bX_1, \dots, \bX_n$ are the leaves of an information flow tree, and is in fact slightly stronger in that the bound is $O(1/t)$ for all~$t$, independently of $\bX_1, \dots, \bX_n$.

In Sections~\ref{sec:prelims}--\ref{sec:star} we will give some results in the direction of proving Conjecture~B; however, we are still unable to prove the conjecture.  The main theorem that we \emph{do} prove in this work is that Conjecture~B holds for \emph{caterpillars}.
\begin{named}{Theorem C}
    Conjecture~B holds when the underlying tree of~$\calT$ is a caterpillar.
\end{named}
Here we are using the following standard graph-theoretic definition:
\begin{definition}
    A \emph{caterpillar graph} is a tree in which every vertex has distance at most~$1$ from a central \emph{spine} (path).  Equivalently, a caterpillar is a graph of pathwidth~$1$.  An example of a caterpillar tree is depicted in Figure~\ref{fig:exampleCat}.
\end{definition}
\begin{figure}
\begin{center}
\begin{tikzpicture}[shorten >=1pt, auto, node distance=0.5cm]
\tikzset{every state/.style={minimum size=0pt}}

	\node[state]		(S1)						{};
	\node[state]		(S2)		[right =1cm of S1]	{};
	\node[state]		(S3)		[right =1cm of S2]	{};
	\node[state]		(S4)		[right =1.5cm of S3]	{};
	\node[state]		(S5)		[right =1cm of S4]	{};
	\node[state]		(S6)		[right =1.2cm of S5]	{};
	\node[state]		(S7)		[right=1cm of S6]	{};

	\node[state] 		(L11)		[below left =1.1cm and 0.15cm of S1]		{};
	\node[state]		(L12)		[below right = 1.1cm and 0.15cm of S1]		{};
	\node[state]		(L21)		[below =1cm of S2]		{};
	\node[state]		(L31)		[below left =1.1cm and 0.4cm of S3]		{};
	\node[state]		(L32)		[below =1cm of S3]		{};
	\node[state]		(L33)		[below right = 1.1cm and 0.4cm of S3]		{};
	\node[state] 		(L41)		[below left =1.1cm and 0.15cm of S4]		{};
	\node[state]		(L42)		[below right = 1.1cm and 0.15cm of S4]		{};
	\node[state]		(L51)		[below =1cm of S5]		{};
	\node[state] 		(L61)		[below left =1.1cm and 0.5cm of S6]		{};
	\node[state]		(L63)		[below right = 1.1cm and 0cm of S6]		{};
	\node[state] 		(L62)		[below left =1.1cm and 0cm of S6]		{};
	\node[state]		(L64)		[below right = 1.1cm and 0.5cm of S6]		{};
	\node[state]		(L71)		[below = 1cm  of S7]		{};

	\path[-] 	(S1)		edge		[above]			(S2)
			(S2)		edge		[above]			(S3)
			(S3)		edge		[above]			(S4)
			(S4)		edge		[above]			(S5)
			(S5)		edge		[above]			(S6)
			(S6)		edge		[above]			(S7)
			
			(S1)		edge		[left]				 (L11)
			(S1)		edge		[right]			 (L12)
			
			(S2)		edge						 (L21)
			
			(S3)		edge		[left]				 (L31)
			(S3)		edge		[right]			 (L32)
			(S3)		edge		[right]			 (L33)
			
			(S4)		edge		[above]			(L41)
			(S4)		edge		[above]			(L42)
			(S5)		edge		[above]			(L51)
			
			(S6)		edge		[above]			(L61)
			(S6)		edge		[above]			(L62)
			(S6)		edge		[above]			(L63)
			(S6)		edge		[above]			(L64)
			
			(S7)		edge		[above]			(L71)
	;

\end{tikzpicture}
\end{center}
\caption{An example of a caterpillar tree}
\label{fig:exampleCat}
\end{figure}
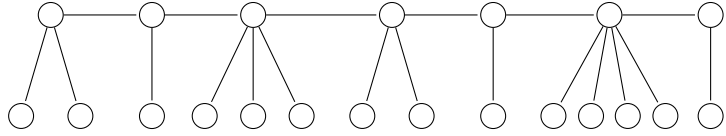
We remark that caterpillar graphs arise quite naturally in some of the contexts where information flow trees are studied; for example, in Hidden Markov Models, where the leaf random variables are observed and the spine random variables are hidden.

\subsection{Motivation and previous work}                   \label{sec:motivation}

Besides being a natural problem in information theory, Conjecture~A is motivated by certain problems in the algorithmic theory of convex relaxation hierarchies.  We give here a very high-level sketch of the connection, as developed in the following works: \cite{Guruswami-Sinop,Barak-Raghavendra-Steurer,Raghavendra-Tan,Barak-Brandao-Harrow-Kelner-Steurer-Zhou,Austrin-Benabbas-Georgiou,Yoshida-Zhou,Barak-Kelner-Steurer,Rothvoss}.

Consider a Boolean optimization problem such as Max-Cut on a graph $G = (V,E)$, where we write $V = [n]$; the task is to find a $\pm 1$ assignment $x_1, \dots, x_n$ to the vertices so as to minimize $\avg_{(u,v) \in E} x_u x_v$.  This is a non-convex (and $\mathsf{NP}$-hard) optimization problem. A natural algorithmic approach is to relax it to an (efficiently-solvable) convex optimization problem and then argue that the relaxed solution can be ``rounded'' to a genuine $\pm 1$ assignment with approximately the same value.  Two important families of such relaxations are the Sherali--Adams LP relaxation and the SOS (Lasserre--Parrilo) SDP relaxation.  The families have a tunable ``degree'' parameter~$t \in \Z^+$; as~$t$ increases, the convex relaxations become tighter and tighter but the running time for solving them increases like $n^{O(t)}$.

Roughly speaking, solving these relaxations yields an optimal solution to the original Max-Cut problem, except that instead of getting a $\pm 1$ assignment $x_1, \dots, x_n$, one gets a collection of ``fake degree-$t$ $\pm 1$-valued random variables'' $\bX_1, \dots, \bX_n$.  In fact, these are not random variables at all; they are merely a list of numbers $\rho_S$ for all $S \subseteq [n]$ with $|S| \leq t$.  However, there is a promise that for each such~$S$ there exists a collection of true $\pm 1$-valued random variables $(\bY_v)_{v \in S}$ with $\E[\prod_{v \in S} \bY_v] = \rho_s$.  Thus, being very imprecise, an algorithm can act as though it has true random variables $\bX_1, \dots, \bX_n$, as long as it only ever uses them in groups of at most~$t$.

The objective function minimized by the convex relaxation is $\alpha \coloneqq \avg_{(u,v) \in E} \rho_{\{u,v\}}$.  An algorithm would now like to take the fake random variables and produce a genuine $\pm 1$ assignment $x_1, \dots, x_n$ which has, say, $\avg_{(u,v) \in E} x_ux_v \leq \alpha + \eps$.  A simple idea for doing this is to draw~$x_j$ according to $\bX_j$, independently for each $j \in [n]$.  (This counts as using the fake random variables in groups of size~$1$ and is thus legal since $t \geq 1$.)  However in doing this we will get $\E[x_u x_v] = \E[\bX_u]\E[\bX_v] = \rho_{\{u\}} \rho_{\{v\}}$, which need not bear any relationship to the quantities $\rho_{\{u,v\}}$ entering into the definition of~$\alpha$.  What would be desirable is if we had $|\rho_{\{u,v\}} - \rho_{\{u\}}\rho_{\{v\}}| \leq \eps$ for all pairs~$(u,v)$, or at least on average over all pairs.  In other words, we wish for the ``average covariance'' (as defined at the beginning of Section~\ref{sec:intro}) of the fake random variables $\bX_1, \dots, \bX_n$ to be smaller than some~$\eps$.  Of course it need not be, but Conjecture~A implies  that it can be made so, provided we are allowed to condition on some $t \leq O(1/\eps)$ randomly chosen $\bX_j$'s.  In the end, using the Sherali--Adams or SOS relaxations with degree parameter~$t$ would allow us to do this in time~$n^{O(1/\eps)}$.

Thus we see that the quantitative dependence in Conjecture~A directly relates to the running time of algorithms based on ``correlation rounding'' of Sherali--Adams/SOS hierarchies.  An example consequence of Conjecture~A (see~\cite{Yoshida-Zhou}) would be that the Sherali--Adams LP hierarchy provides an arbitrarily good multiplicative approximation to Max-Cut on $n$-vertex, $\eps n^2$-edge graphs in time~$n^{O(1/\eps)}$. This gives a very nice tradeoff between density and running time, one that works almost all the way down to the ``sparse'' regime (i.e., $O(n)$ edges). On the other hand, using the weaker Theorem~\ref{thm:r-t}, the running time becomes $n^{O(1/\eps^2)}$.  This is only nontrivial when $\eps \gg n^{-1/2}$; i.e., for graphs with $\omega(n^{3/2})$ edges.

We end this section by commenting on the Raghavendra--Tan proof~\cite{Raghavendra-Tan} of Theorem~\ref{thm:r-t}.  They study the analog $\avgInfoCond{t}(\bX)$ of $\avgCovCond{t}(\bX)$, in which $\abs{\Cov(\bX_u,\bX_v)}$ is replaced by the \emph{mutual information}, $I(\bX_u; \bX_v) \geq 0$.  They deduce very simply from the definitions that for any $0 < T < n-1$,
\[
    \sum_{t = 0}^{T-1} \avgInfoCond{t}(\bX) \leq 1.
\]
This means that there exists a $t < T$ such that $\avgInfoCond{t}(\bX) \leq 1/T$.  The basic relationship $\abs{\Cov[\bX_u,\bX_v]} \leq \sqrt{2} \sqrt{I(\bX_u; \bX_v)}$ lets them complete the proof Theorem~\ref{thm:r-t} with a bound of $t \leq 2/\eps^2$. Thus we see that proving Conjecture~A requires surmounting a familiar difficulty: the quadratic relationship between $L_1$-distance and KL-distance.

Finally, while it's tempting to think that $\avgCovCond{t}(\bX)$ and $\avgInfoCond{t}(\bX)$  should be decreasing functions of~$t$ (thereby allowing us to fix~$t$ independently of $\bX_1, \dots, \bX_n$ in Theorem~\ref{thm:r-t} and Conjecture~A), this is not the case.
\begin{proposition}                                     \label{prop:annoying-counterexample}
    For any fixed integer $T \in \Z^+$, there exist random variables $\bX = (\bX_1, \dots, \bX_{n})$, $n = T+2$, such that $\avgCovCond{t}(\bX) = 0$ for $t < T$ but $\avgCovCond{T}(\bX) = 1$ (and similarly for $\avgInfoCond{t}$).
\end{proposition}
\begin{proof}
    We simply define $\bX_1, \dots, \bX_{T+2}$ to be uniformly random conditioned on $\bX_1 \bX_2 \cdots \bX_{T+2} = 1$.  Then consider any $J \subset [T+2]$ and any outcome of $(\bX_j)_{j \in J}$. If $\abs{J} < T$ then the remaining $\bX_k$'s are (conditionally) pairwise independent. On the other hand, if $\abs{J} = T$ then the remaining pair $(\bX_u,\bX_v)$ is either uniform on $\{(+1,+1), (-1,-1)\}$ or uniform on $\{(+1,-1), (-1,+1)\}$; in either case, the (conditional) covariance is~$1$.
\end{proof}

\subsection{Organization of this paper}
In Section~\ref{sec:prelims}, we describe some basic transformations on information flow trees that preserve the joint distribution on the leaf random variables. These allow us to make certain convenient assumptions about the structure of our information flow trees in in subsequent sections. 
Section~\ref{sec:exp-cov-formula} contains an explicit formula for the covariance of two leaves in an information flow tree conditioned on some outcome of the other leaves.
In Section~\ref{sec:monotonicity} we demonstrate that this expression is nondecreasing as a function of edge correlations along the spine.  This essentially lets us reduce a caterpillar tree to an inhomogeneous star; we analyze the latter in Section~\ref{sec:star}. Finally, the proof of Theorem~C is given in Section~\ref{sec:final}.

\section{Information flow tree equivalences}                          \label{sec:prelims}

Given an information flow tree, there are several ways it can be modified so that the joint distribution of its leaf random variables does not change.  Since Conjecture~B and Theorem~C are only concerned with the leaf random variables, and not the ``internal'' random variables, we are free to make such modifications.  We use the following definition:
\begin{definition}
    Let $\calT$ and $\calT'$ be information flow trees, generating random variables $(\bX_v)_{v \in V}$ and $(\bX'_{v'})_{v' \in V'}$.  Further, assume that~$V$ and~$V'$ have the same set of leaves,~$L$ (though $\calT$ and $\calT'$ may otherwise have different tree topologies and correlation functions).  We say $\calT$ and $\calT'$ are \emph{equivalent} if $(\bX_\ell)_{\ell \in L}$ and $(\bX'_{\ell})_{\ell \in L}$ have the same joint distribution.
\end{definition}
\noindent In this section we describe some transformations on general information flow trees that put them into simpler, equivalent forms.

The first two transformations allow us to assume without loss of generality that $\rho(e) \geq 0$ for all edges~$e$, except possibly for edges incident to leaves.  (In fact, allowing correlations in $[-1,0)$ is not really an essential aspect of our model; the reader will not lose much by simply assuming $\rho \geq 0$ always.)
\begin{lemma}                                       \label{lem:push-negations1}
    Let $\calT = (V,E,\rho)$ be an information flow tree and let $w \in V$ be an internal vertex.  Let $\calT' = (V,E,\rho')$ be the information flow tree that is the same as~$\calT$ except with $\rho'(e) = -\rho(e)$ for all edges~$e$ incident on~$w$.  Then $\calT$ and $\calT'$ are equivalent.
\end{lemma}
\begin{proof}
    Let $(\bX_v)_{v \in V}$ and $(\bR_e)_{e \in E}$ be the random variables generated by~$\calT$.  Define:
    \[
        \bR'_e = \begin{cases}
                     -\bR_e & \text{if $e$ is incident to~$w$,} \\
                     \phantom{-}\bR_e & \text{otherwise;}
                 \end{cases}
        \quad
        \bX'_v = \begin{cases}
                     -\bX_v & \text{if $v = w$,} \\
                     \phantom{-}\bX_v & \text{otherwise.}
                 \end{cases}
    \]
    It's easy to see that $(\bR'_{e})_{e \in E}$ has the correct joint distribution for~$\calT'$.  It's then easy to see that $(\bX'_v)_{v \in V}$ has the correct joint distribution for~$\calT'$. Since~$w$ is not a leaf, we have $\bX_\ell = \bX'_\ell$ for all leaves~$\ell$.  Thus $\calT$ and $\calT'$ are equivalent.
\end{proof}

\begin{lemma}                                       \label{lem:push-negations2}
    For every information flow tree~$\calT = (V,E,\rho)$, there is an equivalent one~$\calT' = (V,E,\rho')$ in which $\rho'(e) \geq 0$ for all ``internal'' edges~$e$; i.e., for edges~$e$ not touching a leaf.
\end{lemma}
\begin{proof}
    Given~$\calT$, choose a root vertex $r \in V$ arbitrarily.  The idea is that in a top-down fashion starting from~$r$, we ``fix'' all negative internal edges.  Specifically, we apply the following procedure:\\

    \hspace*{1in}  for $j = 1, 2, 3, \dots$,

    \hspace*{1in} \quad for each non-leaf vertex~$w$ at distance~$j$ from~$r$,

    \hspace*{1in} \quad \quad if the parent edge~$e$ of~$w$ has $\rho(e) < 0$,

    \hspace*{1in} \quad \quad \quad apply the transformation from Lemma~\ref{lem:push-negations2} to~$w$.\\

    It is easy to see that this procedure terminates with an equivalent information flow tree~$\calT'$ in which all internal edges have a nonnegative correlation value.
\end{proof}

Correctness of the next two transformations follows from the fact that if $\bR_1$, $\bR_2$ are independent $\{\pm 1\}$-valued random variables with expectations $\rho_1, \rho_2$, then $\bR_1\bR_2$ is a $\{\pm 1\}$-valued random variable with expectation $\rho_1\rho_2$:
\begin{lemma}                                       \label{lem:merge}
    Suppose $\calT = (V,E,\rho)$ is an information flow tree, $v \in V$ has degree~$2$, and $(e_1,e_2)$ is the length-two path of edges through~$v$.  Modify~$\calT$ by deleting $v$ and replacing $(e_1,e_2)$ with a single edge~$e$ satisfying $\rho(e) = \rho(e_1)\rho(e_2)$.  Then the resulting information flow tree is equivalent to the original~$\calT$.
\end{lemma}
\begin{lemma}                                       \label{lem:split}
    Let $\calT = (V,E,\rho)$ be an information flow tree and let $e \in E$.  Modify $\calT$ by splitting~$e$ into a length-two path $(e_1,e_2)$ with $\rho(e_1)\rho(e_2) = \rho(e)$.  Then the resulting information flow tree is equivalent to the original~$\calT$.
\end{lemma}
The next two transformations are similar and use the fact that along a path in which all correlations are~$1$, the vertex random variables are always equal.
\begin{lemma}                                       \label{lem:glue}
    Let $\calT = (V,E,\rho)$ be an information flow tree and let $P = (V',E')$ be a connected subgraph of~$(V,E)$ in which $\rho(e) = 1$ for all $e \in E'$.  (For us, $P$ will typically be a path.)  Assume~$V'$ does not contain leaves of~$V$.  Then the information flow tree gotten from~$\calT$ by contracting~$V'$ into a single vertex is equivalent to~$\calT$.
\end{lemma}
\begin{lemma}                                       \label{lem:vertex-split}
    Let $\calT = (V,E,\rho)$ be an information flow tree and let $v \in V$.  For any $m \in \Z^+$, suppose we delete~$v$ and replace it with a path $(w_1, \dots, w_m)$ whose edges are assigned correlation~$1$ by~$\rho$.  For each edge $e = (u,v)$ formerly attached to~$v$, we replace it with $e = (u,w_i)$ for an arbitrarily chosen $i \in [m]$.  Then the resulting information flow tree is equivalent to the original~$\calT$.  An example of such a transformation is depicted in Figure~\ref{fig:splitvertex}.
\end{lemma}

\newcommand{\firstcolor}{ForestGreen}
\newcommand{\secondcolor}{OrangeRed}
\newcommand{\thirdcolor}{Cerulean}

\begin{figure}
\begin{center}
\begin{tikzpicture}[ shorten >=1pt,auto,node distance=2cm]
	{\color{\firstcolor}
	\node[state]		(S1)						{$v_1$};
	}
	{\color{\secondcolor}
	\node[state]		(S2)		[right =2.5cm of S1]	{$v_2$};
	}
	{\color{\thirdcolor}
	\node[state]		(S3)		[right =3cm of S2]	{$v_3$};
	}

	{\color{\firstcolor}
	\node[state] 		(L11)		[below left =1.5cm and 0.3cm of S1]		{$\ell_{1}$};
	\node[state]		(L12)		[below right = 1.5cm and 0.3cm of S1]		{$\ell_{2}$};
	}
	{\color{\secondcolor}
	\node[state]		(L21)		[below =1.2cm of S2]		{$\ell_{3}$};
	}
	{\color{\thirdcolor}
	\node[state]		(L31)		[below left =1.5cm and 0.8cm of S3]		{$\ell_{4}$};
	\node[state]		(L32)		[below =1.2cm of S3]		{$\ell_{5}$};
	\node[state]		(L33)		[below right = 1.5cm and 1.5cm of S3]		{$\ell_{6}$};
	}

	\path[-] 	(S1)		edge		[above]			node{$\rho_1$}		(S2)
			(S2)		edge		[above]			node{$\rho_2$}		(S3)
			
			(S1)		edge		[left]				node{$\rho_3$} (L11)
			(S1)		edge		[right]			node{$\rho_4$} (L12)
			
			(S2)		edge						node{$\rho_5$} (L21)
			
			(S3)		edge		[left]				node{$\rho_6$} (L31)
			(S3)		edge		[right]			node{$\rho_7$} (L32)
			(S3)		edge		[right]			node{$\rho_8$} (L33)

	;

\end{tikzpicture}
{\Large$$\Downarrow$$}
\begin{tikzpicture}[shorten >=1pt,auto,node distance=2cm]
	{\color{\firstcolor}
	\node[state]		(V1)						{$w_{1_1}$};
	\node[state]		(V2)		[right of =V1]	{$w_{1_2}$};
	}
	{\color{\secondcolor}
	\node[state]		(V3)		[right of =V2]	{$w_{2_1}$};
	}
	{\color{\thirdcolor}
	\node[state]		(V4)		[right of =V3]	{$w_{3_1}$};
	\node[state]		(V5)		[right of =V4]	{$w_{3_2}$};
	\node[state]		(V6)		[right of =V5]	{$w_{3_3}$};
	}

	{\color{\firstcolor}
	\node[state] 		(L1)		[below of = V1]		{$\ell_{1}$};
	\node[state]		(L2)		[below of = V2]		{$\ell_{2}$};
	}
	{\color{\secondcolor}
	\node[state]		(L3)		[below of = V3]		{$\ell_{3}$};
	}
	{\color{\thirdcolor}
	\node[state]		(L4)		[below of = V4]		{$\ell_{4}$};
	\node[state]		(L5)		[below of = V5]		{$\ell_{5}$};
	\node[state]		(L6)		[below of = V6]		{$\ell_{6}$};
	}

	\path[-] 	(V1)		edge		[above]			node{$1$}			(V2)
			(V2)		edge		[above]			node{$ \rho_1$}		(V3)
			(V3)		edge		[above]			node{$\rho_2$}		(V4)
			(V4)		edge		[above]			node{$1$}			(V5)
			(V5)		edge		[above]			node{$1$}			(V6)
			
			(V1)		edge						node{$\rho_3$} 	(L1)
			(V2)		edge						node{$\rho_4$ } 	(L2)
			(V3)		edge						node{$\rho_5$ } 	(L3)
			(V4)		edge						node{$\rho_6$ } 	(L4)
			(V5)		edge						node{$\rho_7$ } 	(L5)
			(V6)		edge						node{$\rho_7$ } 	(L6)

	;
\end{tikzpicture}\\

\end{center}
\caption{Applying the transformation of Lemma~\ref{lem:vertex-split} to $v_1$ and $v_3$}
\label{fig:splitvertex}
\end{figure}
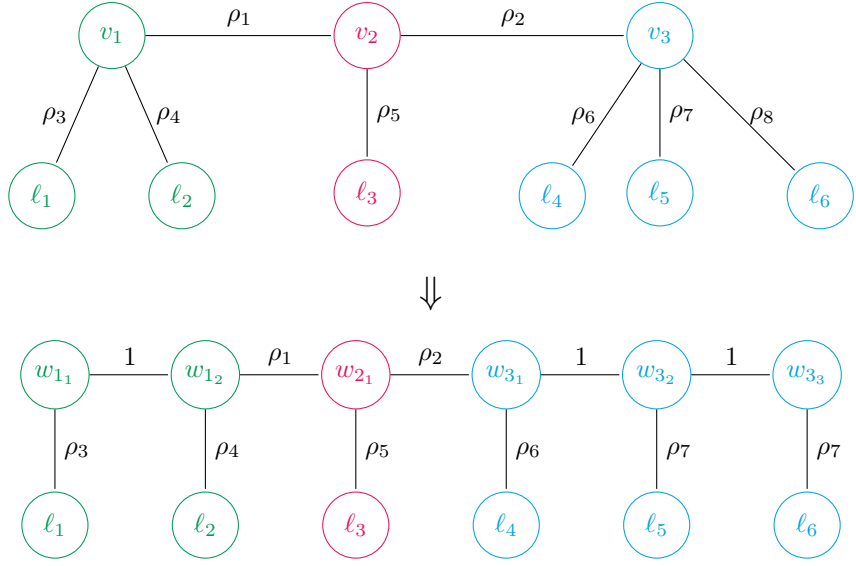

The next transformation allows us to convert to trees of maximum degree~$3$.  Indeed, we can say slightly more.
\begin{lemma}                                       \label{lem:degree-3-ize}
    For each information flow tree $\calT$, there is an equivalent one $\calT'$ in which the underlying graph has maximum degree~$3$.  Indeed, we can take $\calT'$ to be a rooted binary tree in which each internal node has exactly~$2$ children.
\end{lemma}
\begin{proof}
    Suppose $v$ is a vertex in $\calT$ of degree $d > 3$.  We apply Lemma~\ref{lem:vertex-split} to~$v$; by taking $m = d$ when doing so, we have room to attach each of $v$'s neighbors to a different~$w_i$.  As a result, each $w_i$ will have degree at most~$3$.  Repeating this for each vertex of degree exceeding~$3$, we get an equivalent tree~$\calT'$ of degree at most~$3$. By applying Lemma~\ref{lem:split} to an arbitrary edge~$e$, we can root the tree at newly created vertex.  Finally, any vertices of degree~$2$ (other than the root) that remain can be eliminated using Lemma~\ref{lem:merge}.
\end{proof}
Finally, we will use the last lemma to simplify general caterpillar trees.
\begin{definition}                                  \label{def:simple-caterpillar}
    We say a caterpillar is \emph{simple} if the spine has at least two vertices, and each vertex of the spine is attached to exactly one leaf.
\end{definition}
\begin{lemma}                                       \label{lem:simple-caterpillar}
    For each information flow caterpillar~$\calT$, there is an equivalent information flow \emph{simple} caterpillar~$\calT'$.
\end{lemma}
\begin{proof}
    We apply Lemma~\ref{lem:degree-3-ize} to~$\calT$, taking care with one step:  When replacing a high-degree spine vertex~$v$, we insert the new path $(w_1, \dots, w_m)$ as part of the spine, with $v$'s spine neighbor(s) attached appropriately at $w_1$ or $w_m$.  Finally, the resulting binary tree will not quite be a simple caterpillar because the spine node furthest from the root will have two leaf children.  To fix this we can simply take either of these two leaf edges and split it using Lemma~\ref{lem:split}, creating one more spine vertex.
   \end{proof}

\begin{remark}
    It is also possible to convert any information flow tree~$\calT$ into an ``essentially'' equivalent one $\calT' = (V',E',\rho')$ which is \emph{homogeneous} --- meaning $\rho'$ is a constant --- and which has maximum degree~$3$.  Since we won't use this, we merely sketch the conversion.  Given~$\calT$, we can assume it has maximum degree~$3$ using the first part of Lemma~\ref{lem:degree-3-ize}.  Next, fix $\rho' = -(1-\delta)$ for some very small $\delta > 0$.  Now for each edge $e \in E$, replace it with a path of length $k \in \Z^+$ so that $(\rho')^k$ is as close as possible to~$\rho(e)$. Since we can make $\delta$ arbitrarily small, we can get all approximations $(\rho')^k \approx \rho(e)$ simultaneously as close as desired, yielding an ``essentially'' equivalent tree~$\calT'$.  Note that we can't quite ensure all vertices of $\calT'$ have exactly two children: merging the degree-$2$ vertices of~$\calT'$ would spoil its homogeneity property.
\end{remark}

\section{A formula for conditional covariance on general trees}         \label{sec:exp-cov-formula}

Suppose $\bX_{v_1}, \bX_{v_m}$ are two vertex random variables in an information flow tree.  Prior to any conditioning, it's easy to see that their covariance is equal to the product of $\rho(e)$ along the edges~$e$ joining~$v_1$ and~$v_m$.  In this section we generalize this to a formula for their \emph{expected} covariance when conditioned on any event that is comprised of several conditionally independent events.

\newcommand{\Lsub}{L}
\newcommand{\Xbar}{\overline{\bX}}
\newcommand{\bits}{\{\pm 1\}}
\begin{theorem}                                     \label{thm:covariance-induction}
    Let $\calT = (V,E,\rho)$ be a information flow tree, with associated vertex random variables $(\bX_v)_{v \in V}$. Fix any path~$P = (v_1, \dots, v_m)$ of vertices in~$\calT$, where $m \in \Z^+$. We think of~$P$ as partitioning~$\calT$ into a sequence of subtrees~$\calT_i$, with $\calT_i$ rooted at~$v_i$.  For notational simplicity we write
    \[
        \bX_{i} = \bX_{v_i}, \qquad \Xbar = (\bX_1, \dots, \bX_m), \qquad \rho_i = \rho(v_{i},v_{i+1}).
    \]
    Let~$\Lsub_i$ be any event depending only on the random variable outcomes in~$\calT_i$. Write $\Lsub = \Lsub_1 \wedge \Lsub_2 \wedge \dots \wedge \Lsub_m$, and assume~$\Lsub$ has nonzero probability.  Then
    \begin{align}
        \Cov[\bX_{1}, \bX_{m} \mid \Lsub]
		&= \frac{\prod_{i = 1}^{m-1}\rho_i \prod_{i = 1}^m\Pr[L_i \mid \bX_i = +1]\Pr[L_i \mid \bX_i = -1]}{\Pr[L]^2}  \label{eqn:cov-ind1}
    \end{align}
\end{theorem}
\begin{proof}
 Introducing the notation
  \begin{equation}                                \label{eqn:lambda-def}
        \lambda_i^{\pm} = \Pr[\Lsub_i \mid \bX_i = \pm 1],
    \end{equation}
we want to show that
    \begin{equation}                                \label{eqn:cov-ind}
        \Pr[\Lsub]^2 \Cov[\bX_1, \bX_m \mid \Lsub] = \prod_{i = 1}^{m-1} \rho_i \prod_{i = 1}^m \lambda_i^+ \lambda_i^-.
    \end{equation}

    Recall that if $(\bY, \bZ)$ is a pair of random variables, $\Cov[\bY,\bZ] = \frac12 \E[(\bY - \bY')(\bZ - \bZ')]$, where $(\bY',\bZ')$ denotes an independent copy of $(\bY,\bZ)$. Substituting this into~\eqref{eqn:cov-ind}, the identity we want to prove becomes
    \begin{align}
        \prod_{i = 1}^{m-1} \rho_{i} \prod_{i = 1}^m \lambda_i^+ \lambda_i^-
      & = \Pr[\Lsub]^2 \cdot \tfrac12 \E[(\bX_1 - \bX_1')(\bX_m- \bX_m') \mid \Lsub] \nonumber\\
      & = \Pr[\Lsub]^2 \cdot \sum_{x, x' \in \bits^m} \Pr[\Xbar = x \mid \Lsub] \Pr[\Xbar' = x' \mid \Lsub] \cdot \tfrac12(x_1-x_1')(x_m-x_m')
                                                                             \nonumber\\
      & = \sum_{x, x'  \in \bits^m} \Pr[\Xbar = x , \Lsub] \Pr[\Xbar' = x' , \Lsub] \cdot \tfrac12(x_1-x_1')(x_m-x_m').
                                                                             \label{eqn:induct-me}
    \end{align}
    We will prove~\eqref{eqn:induct-me} by induction on~$m$.  The base case, $m = 1$, is
    \begin{equation}                             \label{eqn:cov-base}
        \lambda_1^+ \lambda_1^- = \sum_{x,x' \in \bits} \Pr[\bX_1 = x , \Lsub_1] \Pr[\bX_1' = x' , \Lsub_1] \cdot \tfrac12(x - x')^2.
    \end{equation}
    To verify this, note that when $x = x'$ the summand in~\eqref{eqn:cov-base} is zero and when $x \neq x'$ the summand in~\eqref{eqn:cov-base} is ${2\Pr[\bX_1 = +1 , \Lsub_1] \Pr[\bX_1 = -1 , \Lsub_1]}$.  Thus the whole sum in~\eqref{eqn:cov-base} is indeed
    \begin{multline*}
        4\Pr[\bX_1 = +1 , \Lsub_1] \Pr[\bX_1 = -1 , \Lsub_1] \\ = 4 (\Pr[\Lsub_1 \mid \bX_1 = +1] \Pr[\bX_1 = +1] )(\Pr[\Lsub_1 \mid \bX_1 = -1] \Pr[\bX_1 = -1]) = 4 \lambda_1^+ \cdot \tfrac12 \cdot \lambda_1^- \cdot \tfrac12 = \lambda_1^+ \lambda_1^-.
    \end{multline*}
    We now assume~\eqref{eqn:induct-me} holds for a given~$m \in \Z^+$ and prove it for $m+1$.  Thus we need to show
    \begin{equation}                                   \label{eqn:induction-business}
    \begin{aligned}
        \prod_{i = 1}^{m} \rho_{i} \prod_{i = 1}^{m+1} \lambda_i^+ \lambda_i^-
        = \sum_{\substack{x\phantom{'} \in \bits^m\\ x' \in \bits^m}} \sum_{\substack{x_{m+1} \in \bits \\ x'_{m+1} \in \bits}}
        &\Pr[\Xbar = x,\phantom{'} \bX_{m+1} = x_{m+1}, \Lsub, \Lsub_{m+1}] \cdot {}\\
        &\Pr[\Xbar' = x', \bX'_{m+1} = x'_{m+1}, \Lsub, \Lsub_{m+1}] \cdot \tfrac12(x_1-x_1')(x_{m+1}-x_{m+1}'),
    \end{aligned}
    \end{equation}
    Because of the information flow tree structure we have
    \begin{align*}
        \Pr[\Xbar = x, \bX_{m+1} = x_{m+1}, \Lsub, \Lsub_{m+1}] &= \Pr[\Xbar = x, \Lsub] \Pr[\bX_{m+1} = x_{m+1}, \Lsub_{m+1} \mid \bX_m = x_m] \\
        &= \Pr[\Xbar = x, \Lsub] \cdot (\tfrac12 + \tfrac12 \rho_m x_m x_{m+1}) \cdot \lambda_{m+1}^{x_{m+1}},
    \end{align*}
    and similarly for $x', x'_{m+1}$.  Thus the right-hand side of~\eqref{eqn:induction-business} is
    \begin{equation}                                \label{eqn:inner-mess}
    \begin{aligned}
      \sum_{\substack{x\phantom{'} \in \bits^m \\ x' \in \bits^m}} \Bigl(& \Pr[\Xbar = x, \Lsub]  \Pr[\Xbar' = x', \Lsub]
                                                                        \cdot \tfrac12(x_1 - x_1')  \\
              & \quad {} \cdot \sum_{\substack{x_{m+1} \in \bits \\ x'_{m+1} \in \bits}}
                  (\tfrac12 + \tfrac12 \rho_m x_m x_{m+1})   \cdot \lambda_{m+1}^{x_{m+1}}
            \cdot (\tfrac12 + \tfrac12 \rho_m x_m' x_{m+1}') \cdot \lambda_{m+1}^{x_{m+1}'}
                                                                        \cdot(x_{m+1} - x_{m+1}')\Bigr).
    \end{aligned}
    \end{equation}
    Regarding the inner sum here (i.e., the second line in~\eqref{eqn:inner-mess}), there is no contribution when $x_{m+1} = x_{m+1}'$; by a little algebra, the contribution from the two $x_{m+1} \neq x_{m+1}'$ summands is
    \[
        \tfrac12\lambda_{m+1}^+\lambda_{m+1}^-
            \left( (1+\rho_m x_m)(1-\rho_m x_{m+1}') - (1-\rho_m x_m)(1+\rho_m x_{m+1}')  \right) = \lambda_{m+1}^+\lambda_{m+1}^- \rho_m (x_m - x'_m).
    \]
    Thus~\eqref{eqn:inner-mess} (equivalently, the right-hand side of~\eqref{eqn:induction-business}) is
    \[
    \rho_{m+1} \lambda_{m+1}^+\lambda_{m+1}^-
        \sum_{x,x' \in \bits^m} \Pr[\Xbar = x, \Lsub] \Pr[\Xbar' = x', \Lsub] \cdot \tfrac12(x_1 - x_1')(x_{m} - x_{m}')
        = \prod_{i = 1}^{m} \rho_{i} \prod_{i = 1}^{m+1} \lambda_i^+ \lambda_i^-,
    \]
    by the induction hypothesis~\eqref{eqn:induct-me}.
    \end{proof}
    We present an equivalent way to write~\eqref{eqn:cov-ind1} in the following corollary.
\begin{corollary}
Suppose we are in the setting of Theorem~\ref{thm:covariance-induction}.  Then for any $x \in \{\pm 1 \}^m$ and its bitwise negation $-x$,
 \begin{align*}
        \Cov[\bX_{1}, \bX_{m} \mid \Lsub]
            &= \prod_{i = 1}^{m-1} \rho_i \cdot \frac{\Pr[\Xbar = x \mid \Lsub]}{\Pr[\Xbar = x]} \cdot \frac{\Pr[\Xbar = -x \mid \Lsub]}{\Pr[\Xbar = -x]}.
 \end{align*}
\end{corollary}
\begin{proof}
Beginning with~\eqref{eqn:cov-ind1}, we have
\begin{align*}
\Cov[\bX_{1}, \bX_{m} \mid \Lsub]
		&= \prod_{i = 1}^{m-1}\rho_i \cdot \frac{\prod_{i = 1}^m\Pr[L_i \mid \bX_i = +1]}{\Pr[L]}\cdot\frac{\prod_{i = 1}^m\Pr[L_i \mid \bX_i = -1]}{\Pr[L]}\\
		&= \prod_{i = 1}^{m-1}\rho_i \cdot \frac{\prod_{i = 1}^m\Pr[L_i \mid \bX_i = x_i]}{\Pr[L]}\cdot\frac{\prod_{i = 1}^m\Pr[L_i \mid \bX_i = -x_i]}{\Pr[L]}.
\end{align*}
By virtue of the information flow tree structure we have $\Pr[L_i \mid \Xbar = x] = \Pr[L_i \mid \bX_i = x_i]$.  Thus the above equals
\begin{align*}
		\prod_{i = 1}^{m-1}\rho_i \cdot \frac{\prod_{i = 1}^m\Pr[L_i \mid \Xbar = x]}{\Pr[L]}\cdot\frac{\prod_{i = 1}^m\Pr[L_i \mid \Xbar = -x]}{\Pr[L]}
		= \prod_{i = 1}^{m-1}\rho_i \cdot \frac{\Pr[L \mid \bX = x]}{\Pr[L]}\cdot \frac{\Pr[L \mid \bX = -x]}{\Pr[L]}.
\end{align*}
The proof is now completed by applying Bayes' theorem.
\end{proof}

\section{A monotonicity property}                           \label{sec:monotonicity}

Though the formula in Theorem~\ref{thm:covariance-induction} is quite precise, we will only use it in a rather ``soft'' way, to show a certain monotonicity property.  Suppose we are in the setting of Theorem~\ref{thm:covariance-induction} and that~$\Lsub$ denotes a certain outcome for all the leaf random variables in the tree (these are the events of main interest for us).  Assume for simplicity that the ``path correlations'' $\rho_1, \dots, \rho_{m-1}$ are all nonnegative.  Then the formula from Theorem~\ref{thm:covariance-induction} implies that $\Cov[\bX_1, \bX_m \mid \Lsub]$ is also nonnegative, a fact that does not seem obvious a~priori.  We'll in fact be interested in the \emph{expected value} of this conditional covariance, over all the outcomes~$\Lsub$.

The goal of this section is to show that this expected covariance can only increase if one of the ``path correlations''~$\rho_i$ is increased.  Though this fact seems ``intuitive'', it's also not a~priori obvious; we've only been able to prove it with the aid of Theorem~\ref{thm:covariance-induction}.  Note that this monotonicity property is not immediately obvious from the formula in Theorem~\ref{thm:covariance-induction}, since the expressions $\Pr[\Xbar = (\pm 1, \dots, \pm 1) \mid \Lsub]$ have an implicit dependence on the~$\rho_i$'s.

\begin{theorem}                     \label{thm:monotonic}
    In the setting of Theorem~\ref{thm:covariance-induction}, assume that $\rho_1, \dots, \rho_{m-1} \in [0,1]$.  Let $\overline{\bY} = (\bY_1, \dots, \bY_\ell)$ be the leaf random variables of~$\calT$.  Then
    \begin{equation}                                    \label{eqn:should-increase}
        \E\Brak{\Abs{\Cov[\bX_1,\bX_m]} \mid \bY}
    \end{equation}
    is a nondecreasing function of each~$\rho_i$.
\end{theorem}
\begin{proof}
    Let $y \in \bits^\ell$ be any potential outcome for the leaf random variables, and let $\Lsub_y$ denote the event that $\overline{\bY} = y$.  Then it's easy to see that~$\Lsub_y$ has the factorizable form described in Theorem~\ref{thm:covariance-induction}.  (A slight annoyance is that it's possible to have $\Pr[\Lsub_y] = 0$.  However this can only happen if some pair $\bY_i, \bY_j$ is fully correlated; i.e., $\Cov[\bY_i \bY_j] = \pm 1$.  In this case, letting~$\bX_i$ denote one of the ancestors of $\bY_i, \bY_j$ on path~$P$, we have that~$\bX_i$ is fully determined by \emph{every} possible outcome~$y$. In turn, this means~$\bX_1$ and~$\bX_m$ are \emph{independent} conditioned on every possible outcome~$y$; i.e., the random variable $\Cov[\bX_1,\bX_m \mid \bY]$ is identically~$0$.  Then~\eqref{eqn:should-increase} is trivially a nondecreasing function of the~$\rho_i$'s.  Thus we may henceforth assume that no pair~$\bY_i, \bY_j$ is fully correlated and hence that $\Pr[\Lsub_y] \neq 0$ for all $y \in \bits^\ell$, no matter what the~$\rho_i$'s are.)

    Rewriting identity~\eqref{eqn:cov-ind}, Theorem~\ref{thm:covariance-induction} equivalently states that for any $y \in \bits^\ell$,
    \begin{equation}                                    \label{eqn:the-nonneg-summands}
        \Pr[\Lsub_y] \cdot \Abs{\Cov[\bX_1, \bX_m \mid \Lsub_y]} = \frac{\prod_{i=1}^{m-1} \rho_i \prod_{i=1}^m \lambda_i^+ \lambda_i^-}{\Pr[\Lsub_y]}.
    \end{equation}
    (We are able to insert the absolute-value sign on the left because, as noted earlier, the right-hand side is evidently nonnegative.) By definition, our quantity of interest~\eqref{eqn:should-increase} is the sum of~\eqref{eqn:the-nonneg-summands} over all $y \in \bits^\ell$.  We'll in fact show that for \emph{every} $y \in \bits^\ell$ and every $j \in [m-1]$, the quantity~\eqref{eqn:the-nonneg-summands} is a nondecreasing function of~$\rho_j$.

    In the numerator of~\eqref{eqn:the-nonneg-summands} we have that $\prod_{i \neq j} \rho_i$ is a nonnegative constant independent of~$\rho_j$.  The same is true of $\prod_{i=1}^m \lambda_i^+ \lambda_i^-$: by the definition~\eqref{eqn:lambda-def}, each~$\lambda_i^{\pm 1}$ represents a probability that depends on~$y$ but not on any of the~$\rho_i$'s.  Thus it remains to show that
    \begin{equation}                                    \label{eqn:I-increase}
        \frac{\rho_{j}}{\Pr[\Lsub_y]} = \frac{\rho_j}{\Pr[\overline{\bY} = y]}
    \end{equation}
    is a nondecreasing function of~$\rho_j$. Note that $\Pr[\overline{\bY} = y]$ implicitly depends on all of the~$\rho_i$'s; in fact, it's a \emph{linear} function of each of them.  To see this, note that $\rho_j = \rho(v_j,v_{j+1})$ enters into the generation of $\calT$'s random variables only through the edge random variable~$\bR_{v_j,v_{j+1}}$; thus we can write
    \[
        \Pr[\overline{\bY} = y] = (\half + \half \rho_j) \Pr[\overline{\bY} = y \mid \bR_{v_j,v_{j+1}} = +1] + (\half - \half \rho_j) \Pr[\overline{\bY} = y \mid \bR_{v_j,v_{j+1}} = -1],
    \]
    where the two conditional probabilities on the right do not depend on~$\rho_j$.  Thus we can express~\eqref{eqn:I-increase} as
    \begin{equation}                                    \label{eqn:finish-monotonicity}
        \frac{\rho_j}{\Pr[\overline{\bY} = y]} = \frac{\rho_j}{b + c\rho_j}
    \end{equation}
    for some numbers $b, c$ not depending on~$\rho_j$.  Now a function of this form, $\frac{\rho_j}{b+c\rho_j}$, is nondecreasing if and only if~$b \geq 0$; i.e., if and only if the denominator in~\eqref{eqn:finish-monotonicity} is nonnegative for $\rho_j = 0$.  But indeed this quantity is nonnegative, being a probability.
\end{proof}

We end this section by observing that although we have shown that~\eqref{eqn:should-increase} is an increasing function of the ``path correlations''~$\rho_i$, we actually expect it to be a \emph{decreasing} function of $|\rho(e)|$ for all edges~$e$ \emph{not} on the path between~$\bX_{v_1}$ and~$\bx_{v_m}$.  The intuition is that increasing one such~$|\rho(e)|$ gives more information about its ancestor random variable $\bX_{v_i}$ on the path~$P$.  In turn, this should decrease the expected covariance between $\bX_{v_1}$ and $\bX_{v_m}$.  As an example, if~$v_i$ had just a single edge $(v_i,\ell)$ hanging off it, and $\rho(v_i,\ell)$ were increased to~$1$, then observing the leaf random variable~$\bX_{\ell}$ would determine $\bX_{v_i}$ completely. Thus $\bX_{v_1}$ and $\bX_{v_m}$ would become independent (covariance-$0$) conditioned on any observed outcome for $\bX_\ell$.

\section{The inhomogeneous star}          \label{sec:star}

To motivate the result in this section, let's recall Conjecture~B.  Suppose we are given any information flow tree~$\calT$ and we would like to upper-bound the expected covariance of some \emph{particular} pair of leaves $\bY_u, \bY_v$.  As we'll see, it's easy to reduce this to analyzing the expected covariance of the leaves' parents, call them $\bX_{v_1}, \bX_{v_m}$.  Next, our monotonicity result Theorem~\ref{thm:monotonic} implies that this expected covariance can only increase if all edge-correlations along the path between $\bX_{v_1}, \bX_{v_m}$ were increased to~$1$.  In this case, by Lemma~\ref{lem:glue} we could equivalently think of the entire path as being contracted into one internal random variable~$\bX_0$.

Suppose now that the original tree was a caterpillar---in fact, by Lemma~\ref{lem:simple-caterpillar} we can assume it was a simple caterpillar.  After contracting the path, the collection~$\calL$ of leaves that were originally ``between''~$\bY_u$ and~$\bY_v$ now hang directly off of~$\bX_0$.  The two parts of the caterpillar ``to the outside'' of $\bX_{v_1}$ and~$\bX_{v_m}$ also hang off of~$\bX_0$ as gangly caterpillar-subtrees --- but we plan on ignoring them.  We only intend to analyze the ``inhomogeneous star'' formed by~$\bX_0$ and~$\calL$.  The hope will be that if there is ``squared correlation'' along the edges to~$\calL$, then conditioning on them will typically leave~$\bX_0$ with very small variance; equivalently,~$\bX_{v_1}$ and~$\bX_{v_m}$ will have very small covariance.

The following lemma concerning the inhomogeneous star uses well-known ideas, but as we do not have a reference for the exact statement, we give a proof.
\begin{lemma}                                       \label{lem:inhom-star}
    Let $\calT$ be an information flow tree comprising a ``star center'' vertex with random variable~$\bX_0$, as well as~$m$ leaf vertices with random variables denoted $\bY_1, \dots, \bY_m$.  We allow $\calT$ to contain additional vertices not mentioned.  Write~$\rho_i \in [-1,1]$ for the correlation between~$\bX_0$ and~$\bY_i$, and write $\alpha = \sum_{i=1}^m \rho_i^2$.  Then
    \[
        \E\BRAK{\Var[\bX_0] \mid \overline{\bY}} \leq 4\exp(-\alpha /2),
    \]
    where $\overline{\bY}$ denotes $(\bY_1, \dots, \bY_m)$.
\end{lemma}
\begin{proof}
    Let us define the random variable
    \[
        \bS = \bS(\overline{\bY}) = \sgn(\rho_1 \bY_1 + \cdots + \rho_m \bY_m).
    \]
    (Take $\sgn(0) = +1$ for definiteness.) For each $y \in \bits^m$ let's write
    \[
        p(y) = \Pr[\bX_0 \neq \bS \mid \overline{\bY} = y].
    \]
    Once we condition on~$\overline{\bY} = y$, the random variable~$\bS$ becomes some fixed sign $s \in \bits$, and the random variable~$\bX_0$ takes on some conditioned distribution, call it~$\bZ$.  Now since $\bZ$ is a $\{\pm 1\}$-valued random variable we have
    \[
        \Var[\bZ] = 4\Pr[\bZ = +1]\Pr[\bZ = -1] \leq 4 \Pr[\bZ \neq s],
    \]
    no matter what~$s$ is.  Thus
    \[
        \Var[\bX_0 \mid \overline{\bY} = y] \leq 4p(y),
    \]
    and so
    \[
        \E\BRAK{\Var[\bX_0] \mid \overline{\bY}} \leq 4\E[p(\overline{\bY})] = 4\Pr[\bX_0 \neq \bS].
    \]
    It thus remains to show that
    \[
        \exp(-\alpha/2) \geq \Pr[\bX_0 \neq \bS] \geq \Pr[\bX_0 (\rho_1 \bY_1 + \cdots + \rho_m \bY_m) \leq 0].
    \]
    The two cases $\bX_0 = \pm 1$ are symmetric, so we may assume $\bX_0 = +1$. Then $\bY_1, \dots, \bY_m$ are independent $\{\pm 1\}$-valued random variables with $\E[\bY_i] = \rho_i$, and we wish to show that
    \[
        \Pr[\rho_1 \bY_1 + \cdots + \rho_m \bY_m \leq 0] \leq \exp(-\alpha/2).
    \]
    This follows immediately from Hoeffding's inequality, applied to the  random variables $\rho_i \bY_i \in [-\rho_i, \rho_i]$.
\end{proof}

\section{Proof of Theorem C}\label{sec:final}

In this section we prove Theorem~C.  Let $\calT = (V,E,\rho)$ be a information flow caterpillar tree with~$t \geq 2$ leaves.  By Lemma~\ref{lem:simple-caterpillar} we may assume $\calT$ is a simple caterpillar.  By Lemma~\ref{lem:push-negations2} we may assume that~$\rho$ has a nonnegative value on all spine edges of~$\calT$.  We write $\bX_1, \dots, \bX_t$ for the vertex random variables along~$\calT$'s spine and $\bY_1, \dots, \bY_t$ for the leaf random variables, with $e_i$ denoting the edge between $\bX_i$ and $\bY_i$.  We write $\bR_i = \bX_i\bY_i$ for the edge random variable that $\calT$ associates to $e_i$, and we write $\rho_i = \rho(e_i) = \E[\bR_i]$.
See Figure~\ref{fig:catTree} for a depiction of this.

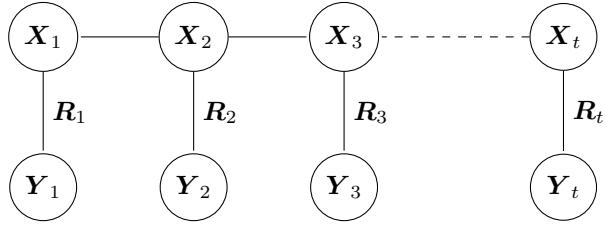
\begin{figure}
\begin{center}
\begin{tikzpicture}[shorten >=1pt,auto,node distance=2cm]
	\node[state] 		(X0)				{$\bX_1$};
	\node[state]		(X1)	[right of=X0]	{$\bX_2$};
	\node[state] 		(Y0) [below of=X0]	{$\bY_1$};		
	\node[state]		(Y1)	[below of=X1]	{$\bY_2$};
	\node[state] 		(X2)	[right of=X1]			{$\bX_3$};
	\node[state]		(Y2)	[below of=X2]	{$\bY_3$};

	\node[state] 		(Xn)		[right =2cm of X2]	{$\bX_t$};
	\node[state] 		(Yn)		[below of=Xn]	{$\bY_t$};

	\path[-] 	(X0)		edge		[right]			node{$\bR_1$}		(Y0)
			(X1)		edge					node{$\bR_2$}		(Y1)
			(X1) 		edge		[above]		node{}		(X0)

			(Xn) 		edge		[dashed, above]				(X2)
			(Xn)		edge		[right]		node{$\bR_t$}		(Yn)
			(X1)		edge					node{}	(X2)
			(X2)		edge					node{$\bR_3$}		(Y2)
				;

\end{tikzpicture}\\
\end{center}
\caption{A Simple Information Flow Caterpillar Tree}
\label{fig:catTree}
\end{figure}

Recall that we wish to show
\begin{equation}                            \label{eqn:conj-cat}
        \avg_{\substack{\text{distinct pairs} \\ u, v \in [t]}} \E\Brak{\Abs{\Cov[\bY_u,\bY_v]} \mid (\bY_j)_{j \in [t] \setminus \{u,v\}}} \leq O(1/t).
\end{equation}
Let us suppose for some time that the pair $u, v \in [t]$ is fixed.  For brevity we'll write $\bcalY = (\bY_j)_{j \in [t] \setminus \{u,v\}}$ for the leaf random variables other than $\bY_u, \bY_v$.  Then
\begin{align}
    \E\Brak{\Abs{\Cov[\bY_u,\bY_v]} \mid \bcalY} = \E\Brak{\Abs{\Cov[\bX_u\bR_u,\bX_v\bR_u]} \mid \bcalY} &= \E\Brak{\Abs{\rho_u\rho_v\Cov[\bX_u,\bX_v]} \mid \bcalY}  \nonumber
    \\ &= \ABS{\rho_u} \ABS{\rho_v} \E\Brak{\Abs{\Cov[\bX_u,\bX_v]} \mid \bcalY} \label{eqn:final-upper-me}
\end{align}
where the second equality uses that $\bR_u, \bR_v$ are independent of $(\bX_u, \bX_v, \bcalY)$.

Given $u, v$, let $\calT'$ denote $\calT$ with edges $e_u, e_v$ deleted.  We may apply our monotonicity result Theorem~\ref{thm:monotonic} to~$\calT'$, with $P$ being the spine path between $\bX_u$ and $\bX_v$.  (Note that we earlier arranged for all spine edges to have nonnegative correlation, as required for Theorem~\ref{thm:monotonic}.)  We conclude that \emph{if} the edge correlations along~$P$ were raised to~$1$, this could only increase the quantity $\E\Brak{\Abs{\Cov[\bX_u,\bX_v]} \mid \bcalY}$ appearing in~\eqref{eqn:final-upper-me}.  We could further upper-bound this quantity as follows:  Write $\calT_{uv}$ for the modification of $\calT'$ in which~$P$ is contracted to a single vertex with random variable called~$\bX_0$ (as in Lemma~\ref{lem:glue}).  Then by applying the inhomogeneous star result, Lemma~\ref{lem:inhom-star} to~$\calT_{uv}$, we would get
\[
    \E\Brak{\Abs{\Cov[\bX_u,\bX_v]} \mid \bcalY} = \E\Brak{\Var[\bX_0] \mid \bcalY} \leq 4\exp(-\alpha(u,v)/2),
\]
where
\[
    \alpha(u,v) \coloneqq \sum \SET{\rho_i^2 : i \text{ is between $u$ and $v$}}.
\]
Putting these observations together, we conclude that for a fixed pair $u, v \in [t]$,
\[
    \E\Brak{\Abs{\Cov[\bY_u,\bY_v]} \mid (\bY_j)_{j \in [t] \setminus \{u,v\}}} \leq \ABS{\rho_u} \ABS{\rho_v} \cdot 4 \exp(-\alpha(u,v)/2).
\]
Thus to complete the proof of~\eqref{eqn:conj-cat} we need to show
\begin{equation}                                    \label{eqn:final-ineq1}
    (*) \coloneqq \avg_{\substack{\text{distinct pairs} \\ u, v \in [t]}} \SET{\ABS{\rho_u} \ABS{\rho_v} \cdot \exp(-\alpha(u,v)/2)} \leq O(1/t).
\end{equation}

This is now simply a combinatorial problem concerning the list of numbers $\rho_1, \cdots, \rho_t$.

We solve the problem as follows.  First, we'd like to switch $u$ and $v$ to being drawn \emph{without} replacement.  Note that
\[
    (*) = \E_{\substack{\bu, \bv \sim [t] \\ \text{uniformly, independently}}}\BRAK{
                \SET{
                \begin{array}{cl}
                    \ABS{\rho_{\bu}} \ABS{\rho_{\bv}} \cdot \exp(-\alpha(\bu,\bv)/2) & \text{if $\bu \neq \bv$,} \\
                    (*) & \text{if $\bu = \bv$.}
                \end{array}
                }}
\]
Since $\ABS{\rho_{\bu}} \ABS{\rho_{\bv}} \cdot \exp(-\alpha(\bu,\bv)/2) \in [0,1]$ always, and since $\Pr[\bu = \bv] = 1/t$, the above differs from
\begin{equation}                                        \label{eqn:final-ineq2}
    \E_{\substack{\bu, \bv \sim [t] \\ \text{uniformly, independently}}} \BRAK{\ABS{\rho_{\bu}} \ABS{\rho_{\bv}} \cdot \exp(-\alpha(\bu,\bv)/2)}
\end{equation}
by at most~$2/t$.  Thus to show~\eqref{eqn:final-ineq1}, it suffices to upper-bound~\eqref{eqn:final-ineq2} by~$O(1/t)$.  To do this, we first apply Cauchy--Schwarz, obtaining
\begin{align}
    \E_{\bu, \bv \sim [t]} \BRAK{\ABS{\rho_{\bu}} \ABS{\rho_{\bv}} \cdot \exp(-\alpha(\bu,\bv)/2)} &\leq \sqrt{\E_{\bu,\bv \sim [t]} \BRAK{\rho_{\bu}^2 \cdot \exp(-\alpha(\bu,\bv)/2)}} \sqrt{\E_{\bu,\bv \sim [t]} \BRAK{\rho_{\bv}^2 \cdot \exp(-\alpha(\bu,\bv)/2)}} \nonumber \\
    &= \E_{\bu \sim [t]} \BRAK{\rho_{\bu}^2 \cdot \exp(-\alpha(\bu,\bv)/2)}, \label{eqn:final-ineq3}
\end{align}
the last equality because $\bu$ and $\bv$ are symmetrically distributed.  Let's introduce the following events:
\[
    A_0 = \text{``$\alpha(\bu,\bv) \in [0,1)$'',} \qquad A_k = \text{``$\alpha(\bu,\bv) \in [2^{k-1},2^k)$''}, \quad k \in \Z^+.
\]
Using the fact that $\sum_{k \geq 0} \bone_{A_k} \equiv 1$, we have that~\eqref{eqn:final-ineq3} equals
\[
    \E_{\bu,\bv \sim [t]} \BRAK{\rho_{\bu}^2 \cdot \sum_{k \geq 0} \bone_{A_k} \cdot \exp(-\alpha(\bu,\bv/2))} \leq \E_{\bu,\bv \sim [t]} \BRAK{\rho_{\bu}^2 \cdot \sum_{k \geq 0} \bone_{A_k} \cdot e^{1/4} \exp(-2^{k-2})}.
\]
Here we essentially lower-bounded $\alpha(\bu,\bv)/2$ by $2^{k-2}$ on the event that~$A_k$ occurs --- except, that is not quite correct when $k = 0$; this why we included the factor $e^{1/4}$, to cover the $k = 0$ case.  Thus it remains to show
\begin{equation}                            \label{eqn:final-ineq4}
    \sum_{k \geq 0} \exp(-2^{k-2}) \cdot \E_{\bu,\bv \sim [t]} \BRAK{\rho_{\bu}^2 \cdot \bone_{A_k}} \leq O(1/t).
\end{equation}
To show this, let's consider a fixed integer $k \geq 0$ and imagine that in the expectation, $\bv \sim [t]$ is chosen first.  Once $\bv$ is chosen, we define interval $U^-_k(\bv) \subseteq [t]$ to be the set of all possible $\bu < \bv$ such that the event $A_k$ occurs.  We define $U^+_k(\bv)$ similarly, but for $\bu > \bv$.  Figure~\ref{fig:depictU} shows a small example.  Denote the union of $U^-_k(\bv)$ and $U^+_k(\bv)$ by $U_k(\bv)$ .

\begin{figure}
\begin{center}

\begin{tikzpicture}[shorten >=1pt, auto, node distance=1cm]
\tikzset{every state/.style={minimum size=0pt}}
	\node[state]		(S0)						{};
	\node[state]		(S1)		[right=0.6cm of S0]	{};
	\node[state]		(S2)		[right =0.6cm of S1]	{};
	\node[state]		(S3)		[right =0.6cm of S2]	{};
	\node[state]		(S4)		[right =0.6cm of S3]	{};
	\node[state]		(S5)		[right =0.6cm of S4]	{};
	\node[state]		(S6)		[right =0.6cm of S5]	{};
	\node[state]		(S7)		[right=0.6cm of S6]	{};
	\node[state]		(S8)		[right=0.6cm of S7]	{};
	\node[state]		(S9)		[right=0.6cm of S8]	{};
	\node[state]		(S10)	[right=0.6cm of S9]	{};
	\node[state]		(S11)	[right=0.6cm of S10]	{};
	\node[state]		(S12)	[right=0.6cm of S11]	{};
	
	\node[state]		(L0)		[below of = S0]		{};
	\node[state]		(L1)		[below of = S1]		{};	
	\node[state]		(L2)		[below of = S2]		{};
	\node[state]		(L3)		[below of = S3]		{};
	\node[state]		(L4)		[below of = S4]		{};
	\node[state]		(L5)		[below of = S5]		{};
	\node[state, fill=black]		(L6)		[below of = S6]		{};
	\node[state]		(L7)		[below of = S7]		{};
	\node[state]		(L8)		[below of = S8]		{};
	\node[state]		(L9)		[below of = S9]		{};
	\node[state]		(L10)		[below of = S10]		{};
	\node[state]		(L11)		[below of = S11]		{};
	\node[state]		(L12)		[below of = S12]		{};

	\path[-]
	(S0) edge (S1)
		edge node [right] {\footnotesize{$.99$}} (L0)
	(S1) edge (S2)
		edge node [right] {\footnotesize{$.96$}} (L1)
	(S2) edge (S3)
		edge node [right] {\footnotesize{$.98$}} (L2)
	(S3) edge (S4)
		edge node [right] {\footnotesize{$.97$}} (L3)
	(S4) edge (S5)
		edge node [right] {\footnotesize{$.99$}} (L4)
	(S5) edge (S6)
		edge node [right] {\footnotesize{$.96$}} (L5)
	(S6) edge (S7)
		edge node [right] {} (L6)
	(S7) edge (S8)
		edge node [left] {\footnotesize{$.99$}} (L7)
	(S8) edge (S9)
		edge node [left] {\footnotesize{$.96$}} (L8)
	(S9) edge (S10)
		edge node [left] {\footnotesize{$.01$}} (L9)
	(S10) edge (S11)
		edge node [left] {\footnotesize{$.97$}} (L10)
	(S11) edge (S12)
		edge node [left] {\footnotesize{$.99$}} (L11)
	(S12) edge node [left] {\footnotesize{$.96$}} (L12);
	
	\draw[decorate,decoration={brace,mirror,raise=6pt}, thick] ($(L5.south west) + (-0.2, 0)$)--($(L6.south west) + (-0.4, 0)$)node [black,midway,yshift=-24pt] {\footnotesize
$U_0^-$};
	\draw[decorate,decoration={brace,mirror,raise=6pt}, thick] ($(L4.south west)+ (-0.2, 0)$)--($(L5.south west)+ (-0.4, 0)$) node [black,midway,yshift=-24pt] {\footnotesize
$U_1^-$};
\draw[decorate,decoration={brace,mirror,raise=6pt}, thick] ($(L2.south west)+ (-0.2, 0)$)--($(L4.south west) + (-0.4, 0)$)node [black,midway,yshift=-24pt] {\footnotesize
$U_2^-$};
\draw[decorate,decoration={brace,mirror,raise=6pt}, thick] ($(L0.south west)+ (-0.2, 0)$)--($(L2.south west)+ (-0.4, 0)$) node [black,midway,yshift=-24pt] {\footnotesize
$U_3^-$};
	\draw[decorate,decoration={brace,mirror,raise=6pt}, thick] ($(L6.south east) + (0.4,0)$)--($(L7.south east) +  (0.2,0)$) node [black,midway,yshift=-24pt] {\footnotesize
$U_0^+$};
	\draw[decorate,decoration={brace,mirror,raise=6pt}, thick] ($(L7.south east)+  (0.4,0)$)--($(L9.south east) +  (0.2,0)$) node [black,midway,yshift=-24pt] {\footnotesize
$U_1^+$};
	\draw[decorate,decoration={brace,mirror,raise=6pt}, thick] ($(L9.south east) +  (0.4,0)$)--($(L11.south east) +  (0.2,0)$) node [black,midway,yshift=-24pt] {\footnotesize
$U_2^+$};
\draw[decorate,decoration={brace,mirror,raise=6pt}, thick] ($(L11.south east) +  (0.4,0)$) --($(L12.south east) +  (0.2,0)$) node [black,midway,yshift=-24pt] {\footnotesize
$U_3^+$};

\draw node [below = 0.15cm of L6] {$\bv$};

\end{tikzpicture}

\end{center}
\caption{A small example illustrating the indices of  $U_k(\bv)$.  The label on edge $e$ is $\rho_e^2$.}
\label{fig:depictU}
\end{figure}
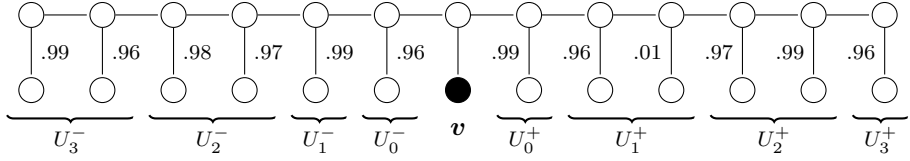
Furthermore, we must have
\[
    \phantom{\quad \forall k \geq 0.} \sum_{\mathclap{u \in U_k(\bv)}} \rho_u^2   = \sum_{\mathclap{u \in U_k^+(\bv)}} \rho_u^2 + \sum_{\mathclap{u \in U_k^-(\bv)}} \rho_u^2  \leq 2^{k} + 2^k  = 2^{k+1} \quad \forall k \geq 0.
\]

It follows that we have the upper bound
\[
    \E_{\bu,\bv \sim [t]} \BRAK{\rho_{\bu}^2 \cdot \bone_{A_k}} = \E_{\bv \sim [t]} \BRAK{\sum_{u \in U_k(\bv)} \Pr[\bu = u] \rho_{u}^2} = (1/t) \E_{\bv \sim [t]} \BRAK{\sum_{u \in U_k(\bv)} \rho_{u}^2} \leq (1/t) \E_{\bv \sim [t]} \BRAK{2^{k+1}} = 2^{k+1}/t.
\]
Substituting this into~\eqref{eqn:final-ineq4}, it remains to observe that indeed
\[
    \sum_{k \geq 0} \exp(-2^{k-2}) \cdot 2^{k+1} \leq O(1).
\]
The proof of Theorem~C is complete.

\section{Conclusions}

Lacking any directions for proving the main Conjecture~A, we believe that Conjecture~B (the case of general information flow trees) is a good place to start.  Having proved Theorem~C (the case of caterpillars), a natural next case to consider is an information flow tree with the property that each leaf is at distance at most \emph{two} from a central spine.  By the transformations in Section~\ref{sec:prelims}, it suffices to consider the case that each spine node has a single edge hanging off it, which in turn has an inhomogeneous star hanging off it.  Perhaps some of the ``reconstruction'' results from~\cite{evans} in terms of effective electrical resistance could be of use here.  Another interesting special case of Conjecture~B that one could try to resolve is that of a complete binary tree in which all edge correlations have the same value~$\rho$.  (This is the most heavily-studied information flow tree.)  We believe that this case satisfies Conjecture~B by a wide margin for all~$\rho$, even with a sub-inverse-polynomial bound in place of~$O(1/t)$.  Perhaps the formula in our Theorem~\ref{thm:covariance-induction} could help prove this.

\subsection*{Acknowledgments}
We thank Yuan Zhou for several early discussions on the topic of this paper.

\bibliographystyle{plain}
\bibliography{cat}
\end{document}